\documentclass{article}

\usepackage{amsmath,amsthm,amssymb}
\usepackage[breaklinks]{hyperref}
\usepackage{algorithm}
\usepackage[noend]{algorithmic}

\DeclareSymbolFont{symbolsC}{U}{txsyc}{m}{n}
\SetSymbolFont{symbolsC}{bold}{U}{txsyc}{bx}{n}
\DeclareFontSubstitution{U}{txsyc}{m}{n}
\DeclareMathSymbol{\coloneqq}{\mathrel}{symbolsC}{66}

\algsetup{indent=2em}

\newcommand\remove[1]{}

\newcommand{\lca}{\mathrm{\bf lca}}

\DeclareMathOperator{\diam}{diam}

\DeclareMathOperator{\Processed}{\mathbf{Processed}}

\newtheorem{theorem}{Theorem}
\newtheorem{lemma}{Lemma}[section]

\theoremstyle{definition}

\newtheorem{definition}[theorem]{Definition}

\begin{document}

\title{Fast C-K-R Partitions of Sparse Graphs\thanks{\noindent M.~Mendel was partially supported by an ISF grant no. 221/07, a BSF grant no. 2006009, and
a gift from Cisco research center. 
This work is part of the M.Sc. thesis of C.~Schwob prepared in the Computer Science Division of the Open University of Israel.
}
}

\author{Manor Mendel\\
Computer Science Division\\
The Open University of Israel\\
\texttt{mendelma@gmail.com} 
\and
{Chaya Schwob}\\
Computer Science Division\\
The Open University of Israel\\
\texttt{cschwob@nds.com}
}

\date{}

\maketitle

\begin{abstract}
We present fast algorithms for constructing probabilistic embeddings and approximate distance oracles in sparse graphs. The main ingredient is a fast algorithm for sampling the probabilistic partitions of Calinescu, Karloff, and Rabani in sparse graphs. 
\end{abstract}

\section{Introduction}\label{section:introduction}

Metric decompositions aim to partition the points of a metric space
into blocks such that close-by points tend to be placed in the same
block while distant pairs of points in different blocks.
For most metric spaces, no straightforward interpretation of
these goals exists.

One successful compromise is the notion of \emph{probabilistic
partition.}
{A } $\Delta$-bounded probabilistic partition is a probability distribution over partitions of the metric space, 
such that in every partition in the distribution, the diameters of the blocks are at most $\Delta$, while ``close-by" pairs of points  are in the same block with ``high" (or at least ``non-negligible") probability.

Probabilistic partitions first appeared, to the best of our knowledge,%
\footnote{Closely related notions of partitions appeared before, e.g. in~\cite{LR99}.}
in a paper of Linial and Saks~\cite{LS94}, 
and publicized in the work of Bartal~\cite{bartal96} on probabilistic embeddings.
Calinescu, Karloff and Rabani~\cite{calinescu01} introduced the following probabilistic partition of metric spaces which we describe as an algorithm that samples a partition from the probability distribution.

\begin{algorithm}[!ht]
\caption{CKR-Partition} \label{alg:ckrpar}
\begin{algorithmic}
\REQUIRE A finite metric space $\left(X, \rho\right)$,  scale $\Delta>0$
\ENSURE Partition $P$ of $X$
\STATE $\pi \coloneqq$ random permutation of $X$
\STATE $R \coloneqq$ random  number in $\left[\frac{\Delta}{4},\frac{\Delta}{2}\right]$ 
\FOR{$i=1$ to $|X|$}
\STATE $C_i\coloneqq \{y\in X: \rho(y,x_{\pi(i)}) \le R \} \setminus \bigcup_{j=1}^{i-1}C_j $
\ENDFOR
\RETURN $P\coloneqq \left\{C_1,\ldots,C_{|X|}\right\}\setminus \left\{\emptyset\right\}$
\end{algorithmic}
\end{algorithm}

We call the probabilistic partition $P$ sampled by Algorithm~\ref{alg:ckrpar}, \emph{$\Delta$-bounded CKR partition}.
Na\"{\i}ve implementations of Algorithm~\ref{alg:ckrpar} take $\Omega(n^2)$ time for $n$-point metric spaces. 
It seems hard to break the $\Omega(n^2)$ barrier  on the running time in general finite metric spaces.
However, in many situations, the metric spaces we deal with come from the shortest-path metric on relatively sparse graphs.
In those cases we can do better, as the following theorem shows.

\begin{theorem} \label{thm:fast-ckr}
 Suppose we are given a positive number $\Delta>0$ and an undirected
  graph with positive edge weights $G=(X,E,\omega)$. Suppose $G$ has
  $n$ vertices and $m$ edges, and let $\rho$ denote the shortest-path metric in $G$.
  One can sample a $\Delta$-bounded CKR partition of $(X,\rho)$ in expected
  $O(m\log n+n \log^2 n)$ time.
  \end{theorem}
The sampling will be accomplished by Algorithm~\ref{alg:graphckr} (Section~\ref{section:CKRrandompartition}).
  
CKR partitions have found many algorithmic (as well as mathematical) applications, 
and we mention only few of them here.
They were introduced as part of an approximation algorithm to the $0$-extension problem~\cite{calinescu01,harrelsonrao03}.
Fakcharoenphol, Rao and Talwar~\cite{raotalwar04} used them to obtain an asymptotically
tight probabilistic embedding into trees, which we call FRT-embedding. Probabilistic embeddings are used in many of the best known approximation and online algorithms as a reduction step from general metrics into tree metrics. 
Mendel and Naor~\cite{mendel06} showed that FRT-embedding possesses  a stronger embedding property, which they called ``maximum gradient embedding".
Recently, R{\"a}cke~\cite{racke2008} used FRT-embeddings to obtain hierarchical decompositions for
congestion minimization in networks, and used them to give an $O(\log n)$
approximation algorithm for the minimum bisection problem and an $O(\log n)$ competitive online algorithm for the oblivious routing problem.
Krauthgamer et.~al~\cite{Krauthgamer2004} used CKR-partitions to give a new proof of Bourgain's embedding theorem.
Mendel and Naor~\cite{mendelnaor07} used them to obtain an asymptotically tight metric Ramsey theorem and approximate distance oracles.

The improved running time of the sampling of CKR partitions may improve the running time of many of their applications. In order to keep the paper short we work out the details of only two (related) applications of CKR  partitions: FRT-embeddings, and approximate distance oracles based on CKR-partitions.

\subsubsection*{Probabilistic embedding into ultrametrics~\cite{AKPW,bartal96}.} 
An ultrametric $\nu$ on $X$ is a metric
which satisfies $\nu(x,z)\le \max\{\nu(x,y), \nu(y,z)\}$, for every $x,y,z\in X$.
A probabilistic embedding of a metric space $(X,\rho)$ into ultrametrics with distortion $D$ is a probability distribution
$\Pi$ over ultramerics $\nu$ on $X$ such that
\begin{enumerate}
\item For every $x,y\in X$, $\Pr_{\Pi}[\nu(x,y)\ge \rho(x,y)]=1$.
\item For every $x,y\in X$, $\mathbb E_{\Pi}[\nu(x,y)] \le D \cdot \rho(x,y)$.
\end{enumerate}

FRT-embedding is a probabilistic embedding into ultrametrics with distortion $O(\log n)$ for every $n$-point metric space~\cite{raotalwar04}.
This bound is asymptotically tight for certain classes of finite metric spaces, such as graphs of high girth\cite{bartal96}, grids~\cite{AKPW}, and expanders~\cite{LLR}.

\subsubsection*{Approximate distance oracles.}
An approximate distance oracle is a data structure with ``compact" ($o(n^2)$) storage that answers (approximate) distance queries in a given $n$-point metric space in constant time. A simple counting argument over
all bi-partite graphs shows that exact, and even $2.99$ approximation is impossible when the storage is $o(n^2)$. The history of this problem is nicely summarized in~\cite{thorupzwick01}. In particular, Thorup and Zwick~\cite{thorupzwick01} gave an asymptotically tight trade-off between the approximation and the storage%
\footnote{The lower bound on the approximation assumes a conjecture of Erd\"os about the number of edges possible in graph with a given number of vertices and a given girth, see~~\cite{thorupzwick01}.}: 
For every $k\in \mathbb N$, they constructed $(2k-1)$-approximate distance oracle requiring $O(k n^{1+1/k})$ storage, and answering queries in $O(k)$ time. Recently Mendel and Naor~\cite{mendelnaor07} presented  different approximate distance oracles, based on CKR partitions. While those oracles do not give   optimal approximation/storage trade-off,%
\footnote{As reported in~\cite{mendelnaor07}, oracles of size $O(n^{1+1/k})$ support approximation factor of $128k$ in the queries. While the constant $128$ can be reduced by optimizing the parameters in the construction, it is unlikely to get below 8.}
they answer distance queries in an absolute constant time, regardless of the approximation parameter. 

\begin{theorem} \label{thm:application-1}
Let $G=(X,E,\omega)$ be an $n$-vertex weighted graph  with $m$ edges, and let $\rho$ be the shortest-path metric on $X$. Then
\begin{enumerate}
\item It is possible to sample from FRT-embedding of $(X,\rho)$ in $O(m \log^3 n)$ expected time.
\item It is possible to construct in $O(m n^{1/k} \log^3 n)$ expected time an $O(k)$-approximate distance oracle for $(X,\rho)$ based on CKR partitions whose storage is $O(n^{1+1/k})$.
\end{enumerate}
\end{theorem}

For approximate distance oracles, it is also possible to improve the na\"{\i}ve construction time even when the metric is given as distance matrix, by first constructing a spanner of the metric with $o(n^{2})$ edges, and then use the fast CKR partitions for sparse graphs on that spanner.

\begin{theorem} \label{thm:hierachical-part}
For $n$-point metric spaces given as distance matrix, it is possible to construct $O(k)$-approximate distance oracle based on CKR partitions {whose storage is} $O(n^{1+1/k})$, in $O(n^2)$ expected time.
\end{theorem}

We remark that a different probabilistic partition, developed by Bartal~\cite{Bartal04,Bartal05} and Abraham et.~al.~\cite{ABN06}, have properties similar to  (and even stronger than) CKR partitions. 
However, we do not see an easy way to quickly obtain a sample from this distribution when the graph is sparse.

\subsection*{Further Results}
The second named author presents in~\cite{mythesis} an efficient PRAM algorithm for sampling CKR partitions and constructing approximate distance oracles in weighted graphs. The running time of the algorithm is $\mathrm{polylog}(n)$ and the total work is $O(m\,\mathrm{polylog}(n))$.

\subsection*{Outline of the paper.}
After setting up in Section~\ref{section:preliminaries} the notation and reviewing
the properties of CKR partitions, we prove Theorem~\ref{thm:fast-ckr} in Section~\ref{section:CKRrandompartition}.

In many applications (and in particular, probabilistic embeddings and approximate distance oracles), probabilistic partitions are applied hierarchically, using an exponentially decreasing series of scales. 
This na\"{\i}vly implies an added $O(\log \Phi)$ factor in the running time, where $\Phi$ is the spread\footnote{The spread is the ratio between the diameter and the smallest non-zero distance in the metric} of the metric.
There is a standard technique that converts this $\log \Phi$ factor into a $\log n$ factor. However, we are not aware of a concrete implementation that satisfies the efficiency requirements needed in this paper. We therefore sketch the details of this technical step in Section~\ref{section:optimization}.
The specific applications examined in this paper, Theorem~\ref{thm:application-1} and Theorem~\ref{thm:hierachical-part}, are discussed in Section~\ref{section:applications}.

\section{Preliminaries}\label{section:preliminaries}

For simplicity of the presentation, the model of computation we assume is a unit-cost, real-word RAM machine.
In this model words can hold real numbers and arithmetic, comparison, and truncation operations take unit time. 
Our algorithms, however, do not take advantage of the unrealistic power of this model, and can also be 
presented in a more realistic computational models such as the "unit cost floating-point word RAM model" (cf.~\cite[Sec.~2.2]{harpeledmendel06}).

The diameter of a finite subset {$Y\subseteq (X,\rho)$} is defined as $\diam(Y)=\max\{\rho(x,y):\; x,y\in Y\}$.
For simplicity of the presentation, we assume that the given finite metric $(X,\rho)$ has minimum non-zero distance 1, and diameter $\diam(X)=\Phi$.

The (closed) ball around $x$ at radius $r$ is defined as $B_\rho(x,r)=\{y\in X:\; \rho(x,y)\le r\}$. When $\rho$ is clear from the context we may omit it from the notation. Given a partition $P$ of $X$, and $x\in X$, we denote
by $P(x)$ the block of $P$ which contains $x$.

$\Delta$-bounded CKR partitions have an obvious upper bound of $\Delta$ on the diameter of the blocks in the partition.
The following is the padding property they enjoy.
\begin{lemma}[\cite{raotalwar04,mendelnaor07}]  \label{lem:improvedCKR}
Let $ P$ be a $\Delta$-bounded CKR partition  of the metric $(X,\rho)$. Then, for every $x\in X$, 
and $t\le \Delta/8$,
\begin{equation}
\label{eq:betterCKR} 
\Pr_{P\sim \mathcal P}\left[B\left(x,t\right)\subseteq 
P(x)\right]\ge
\left(\frac{\left|B_X(x,\Delta/8)\right|}{\left|B_X(x,\Delta)\right|}\right)^{\frac{16t}{\Delta}}.
\end{equation}
\end{lemma}

\medskip 

In this paper we define hierarchical partition of a metric space $(X,\rho)$ as a sequence of $\lceil \log_8 \Phi \rceil +2$ partitions
$P_{-1},P_0,\ldots , P_{\lceil \log_8 \Phi \rceil}$ such that $P_i$ is a partition of $X$ at scale $8^i$,
and $P_{i}$ is a refinement of $P_j$ when $i\le j$, i.e., for every $x\in X$, $P_i(x)\subseteq P_j(x)$.
Given a sequence $(Q_j)_{j\ge -1}$ of partitions, where $Q_j$ is  $8^j$-bounded partition of $X$, the common refinement of $(Q_j)_{j}$ is a hierarchical {partition}
$(P_j)_{j\ge -1}$ in which $P_j=\{\bigcap_{\ell \ge j} C_\ell:\; C_\ell\in Q_\ell\}$

By sampling stochastically independent CKR partitions at the different scales and then taking their common refinement, we obtain the following result.

\begin{lemma}[\cite{mendelnaor07}] \label{lem:ckr-hierarchical-part}
Fix a finite metric space $(X,\rho)$. Then there exists (efficiently {sampleable}) probability distribution $\mathcal H$ over hierarchical partitions such that for every $x\in X$, and every $0<\beta<1/8$,
\[
\Pr_{(P_{-1},\ldots, P_{\lceil \log_8 \Phi \rceil})\sim \mathcal H} \left[\forall\ k\ge -1 ,\ 
B(x,\beta 8^{k})\subseteq P_{k}(x)\right] \ge \left|X\right|^{-16 \beta}.
\]
\end{lemma}

A finite ultrametric $(X,\nu)$ can be represented by {a} tree as follows.
\begin{definition} An ultrametric tree $\left(T, \Gamma\right)$  is a metric space whose elements are the leaves of a rooted finite tree $T$. Associated with every vertex $u\in T$ is a label $\Gamma\left(u\right)\geq 0$ such that $\Gamma\left(u\right) = 0$ iff $u$ is a leaf of $T$. If a vertex $u$ is a child of a vertex $v$ then $\Gamma\left(u\right)\leq \Gamma\left(v\right)$ . The distance between two leaves $x,y\in T$ is defined 
as $\Gamma\left(\lca\left(x,y\right)\right)$, where $\lca\left(x,y\right)$ is the least common ancestor of $x$ and $y$ in T. 
\end{definition}
Every finite ultrametric can be represented by {an} ultrametric tree, and vice versa: the metric on ultrametric tree is a finite ultrametric. 
Hierarchical partition $\{P_k\}_{k=-1}^{\lceil \lg \phi \rceil}$ of $(X, \rho)$ naturally corresponds to an ultrametric  $\nu$ on $X$ where $\nu(x,y) = 8^{\min \left\{ j  : \;  P_j\left(x\right)=P_j\left(y\right) \right\}}$.

\medskip

Let $G=\left(X,E,\omega\right)$ be an undirected positively weighted graph.
Let $\rho: X\times X \to \left[0, \infty\right)$ be the shortest-path metric on $G$.
We denote by $n=|X|$ the number of vertices, and by $m=|E|$ the number of edges. We assume an adjacency list representation of graphs.

The single source shortest paths in weighted undirected graphs problem [USSSP] is used as a subroutine in  our algorithm. Given a weighted graph with $n$ vertices and $m$ edges, Dijkstra's classical USSSP algorithm~\cite{Dijkstra} with source $w$ maintains for each vertex $v$ an upper bound on the distance between $w$ and $v$, $\delta\left( v\right)$. If $\delta\left( v\right)$ has not been assigned yet, it is interpreted as infinite. Initially, we just set $\delta\left( w\right)=0$, and we have no visited vertices. At each iteration, we select an unvisited vertex $u$ with the smallest finite $\delta\left(u\right)$, visit it, and relax all its edges. That is, for each incident edge $\left(u, v\right)\in E$, we set $\delta\left( v\right) \leftarrow \min{\left\{\delta\left( v\right),\delta\left( u\right)+\omega\left(u, v\right)\right\}}$. We continue until no vertex is left unvisited.
Using Fibonacci heaps~\cite{Fredman87} or Bordal's priority queues~\cite{Brodal}, Dijkstra's algorithm is implemented in $O\left(m+n\lg n\right)$ time.

\section{Fast CKR partitions} \label{section:CKRrandompartition}

Given an undirected positively weighted graph $G=\left(X, E,\omega\right)$ with $n$ vertices and $m$ edges  whose shortest path metric is denoted by $\rho$, and $\Delta> 0$,
we show how to implement Algorithm~\ref{alg:ckrpar}  in $O\left(m\lg n+ n\log^2 n\right)$ expected time.

First, we sample a random permutation $\pi$, which can be generated in linear time using several methods, \emph{e.g.,} Knuth Shuffle (see~\cite{bermanklamkin76}). Next, we sample $R$ uniformly%
\footnote{A closer look on the analysis of the CKR partitions (see~\cite{mendelnaor07}) reveals that it is sufficient to sample $R$ from discrete distribution having resolution of $\Delta/c\log n$, and therefore this step can be carried out in a ``realistic" computational model such as the unit cost floating-point word RAM model.}
 in the range $\left[\frac{\Delta}{4},\frac{\Delta}{2}\right]$.

We then use a variant of Dijkstra's algorithm for computing the blocks. The algorithm performs $\left|X\right|$ iterations. In the $i$-th iteration, all vertices in $B_\rho\left(x_{\pi\left(i\right)},R\right)$ not yet assigned to some block are put in $C_i$. 
In order to gain the improved running time of Theorem~\ref{thm:fast-ckr}, 
we change Dijkstra's algorithm to return the distance of a point $v$ from $\pi(i)$
\emph{only if} this distance is smaller then the distance of $v$ from $\pi(j)$ for all $j<i$.

Technically, this is done as follows.
Consider the $i${-}th iteration and let $\delta(v)$ be the variable that holds the Dijkstra's algorithm's current estimate
on the distance between $\pi(i)$ and $v$. In Dijkstra's algorithm $\delta(v)$ is usually initialized to $\infty$ and then gradually decreases until $u$ is extracted from the priority queue, at which point $\delta(v)=\rho(\pi(i),v)$. In the variant of Dijkstra's algorithm used in Algorithm~\ref{alg:graphckr},  $\delta(\cdot)$ are not reinitialized when the value of $i$ is changed. This means that now at the end
of the $(i-1)$-th iteration, $\delta(v)=\min_{j< i} \rho(\pi(j),v)$, 
which in turn means
that an edge $(u,w)$ is relaxed in the $i$-th iteration only when $\pi(i)$ is the closest center to both $u$, and $w$ among $\pi(j)$, $j\le i$. This dramatically reduces the
number of relaxations being done, and does not hurt the correctness of the algorithm.
 The full details are given in Algorithm~\ref{alg:graphckr}.

\begin{algorithm}[ht]
\caption{Graph-CKR-Partition} \label{alg:graphckr}
\begin{algorithmic}[1]
\REQUIRE Graph $G=\left(X, E,\omega\right)$,  scale $\Delta>0$
\ENSURE Partition $P$ of $X$
\STATE Generate random permutation $\pi$ of $X$
\STATE Sample a random $R\in \left[\frac{\Delta}{4},\frac{\Delta}{2}\right]$
\FORALL {$v\in X$}
\STATE $\delta (v)\coloneqq \infty$
\STATE $ P\left(v\right) \coloneqq 0$
\ENDFOR
\FOR [Perform modified Dijkstra's alg starting from $\pi\left(i\right)$] {$i\coloneqq 1$ to $\left|X\right|$} \label{ln:for-i}
\STATE $\delta\left(\pi\left(i\right)\right)\coloneqq 0$
\STATE  $Q \coloneqq \emptyset$  \COMMENT{$Q$ is a priority queue with $\delta$ being the key }
\STATE $w\coloneqq\pi\left(i\right)$
\WHILE[$w$ is visited now] {$\delta\left(w\right)\leq R$} \label{ln:while-w}
\IF {$P\left(w\right)=0$}
\STATE $ P\left(w\right) \coloneqq i$ 
\ENDIF
\FORALL{$u:\ \left(u,w\right)\in E$}
\IF [Relax edges adjacent to $w$] {$\delta(u)>\delta(w)+\omega(u,w)$}
\STATE $\delta\left(u\right) \coloneqq \delta\left( w\right) + \omega\left(u, w\right)$ 
\IF{$u\notin Q$}
  \STATE Insert $u$ into $Q$
\ENDIF
\ENDIF
\ENDFOR
\STATE Extract $w\in Q$ with minimal $\delta\left(w\right)$
\ENDWHILE
\ENDFOR \label{ln:end-for-i}
\RETURN $P$
\end{algorithmic}
\end{algorithm}

\begin{lemma} \label{lem:graph-ckr-is-honest}
After the $i$-th iteration of the loop on lines~\ref{ln:for-i}--\ref{ln:end-for-i} of Algorithm~\ref{alg:graphckr},
\begin{equation} \label{eq:delta-i}
\delta(v)= \begin{cases} \min_{j\le i} \rho(\pi(j),v)\ &  \text{if }\min_{j\le i} \rho(\pi(j),v)\le R, \\
\infty  & \text{otherwise}.
\end{cases}
\end{equation}
\end{lemma}
\begin{proof}[Sketch of a proof]
Proof by induction on $i$. When $i=1$, and as long $\delta(w)\le R$ in the \emph{While} loop of line~\ref{ln:while-w},
the algorithm behaves \emph{exactly} as Dijkstra's algorithm and hence~\eqref{eq:delta-i} is true for $i=1$.

Assume inductively that~\eqref{eq:delta-i} is correct for $i-1$. 
If $ \min_{j\le i-1} \rho(\pi(j),v) \le \rho(\pi(i),v)$, then clearly the $i$-th iteration will not change $\delta(v)$, and by the inductiion hypothesis we are done.

Assume now that $ \min_{j\le i-1} \rho(\pi(j),v) > \rho(\pi(i),v)$, and $\rho(\pi(i),v)\le R$.
Let $\pi(i)=v_0,v_1,\ldots, v_\ell=v$ be a shortest-path between $\pi(i)$ and $v$. We claim that for every $t\in\{1,\ldots, \ell\}$,
$ \min_{j\le i-1} \rho(\pi(j),v_t) > \rho(\pi(i),v_t)$, since otherwise we had
\[ \min_{j\le i-1} \rho(\pi(j),v) \le \min_{j\le i-1} \rho(\pi(j),v_t) +\rho(v_t,v) \le 
\rho(\pi(i),v_t)+\rho(v_t,v)=\rho(\pi(i),v). \]
Hence all the edges along the  path $\pi(i)=v_0,\ldots,v_\ell=v$ will be relaxed in the $i$-th iteration, and so in the end of the $i$-th iteration,
$\delta(v)=\rho(\pi(i),v)$.
\end{proof}

\begin{proof}[Proof of Theorem~\ref{thm:fast-ckr}]
We first prove the correctness of Algorithm~\ref{alg:graphckr}, i.e., that $P\left(v\right) = \min \left\{ i :\ \rho(\pi(i),v)\le R\right\}$ for every $v\in V$.
Let $i_0=\min \left\{ i :\ \rho(\pi(i),v)\le R \right\}$. This means that $\min_{j<i_0} \rho(\pi(j),v) > R \ge \rho(\pi(i_0),v)$. By Lemma~\ref{lem:graph-ckr-is-honest}  at the beginning of the $i_0$-th iteration, $\delta(v)=\infty$, and hence $P(v)=0$ and by the end of the $(i_0)$-th iteration, $\delta(v)=\rho(\pi(i_0),v)$, and necessarily $P(v)=i_0$. Note that once $P(v)$ is set to a non-zero value, its value will not change.

We next bound the running time. we will show that every vertex is inserted into the queue $O(\log n)$ times in expectation, and
every edge $(u,v)$ of $G$ undergoes  $O\left(\log n\right)$ relaxations in expectation. 
Consider the non-increasing sequence $a_i=\min_{j\leq i}\rho\left(\pi\left(j\right), v\right)$. In the $i$-th iteration, $\delta(v)$ decreases if and only if
$a_{i-1}> a_i$. Note that $a_{i-1}>a_i$ means that $\rho\left(\pi\left(i\right),v\right)$ is the minimum among $\left\{\left. \rho\left(\pi\left(j\right),v\right)\ \right|\ j\leq i\right\}$, and the probability (over $\pi$) for this to happen is at most $1/i$.
By linearity of the expectation, the expected number of rounds of the $i$-loop where $\delta(v)$ decreases
(and hence $v$ is inserted into the queue) is at most 
\[ \sum^{n}_{i=1}\frac{1}{i} \le 1+ \ln n. \]
Furthermore, by another application of the linearity of expectation, 
the expected number of edge relaxations is at most
\[O\Bigl(\sum_{v\in V}\ln n\cdot \deg\left(v\right)\Bigr) = O\left(m\log n\right).\]

Using Fibonacci heaps~\cite{Fredman87} or Brodal's priority queues~\cite{Brodal}, 
the total running time of Algorithm~\ref{alg:graphckr} is $O(r+s\log n)$, where
$r$ is the number of relaxations, and $s$ is the number of ``insert" and ``extract minimum" operations. In our cases $\mathbb E[r]=O(m\log n)$, and $\mathbb E[s]=O(n\log n)$. Therefore the total expected running time of Algorithm~\ref{alg:graphckr}
is $O(m \log n +n \log^2 n)$.
\end{proof}

\section{Hierarchical Partitions}\label{section:optimization}


In this section we explain how to dispense with the $O\left(\log \Phi\right)$ factor in the na\"{\i}ve implementation of the hierarchical partitions,  and replace it with $O\left(\log n\right)$.   
The method being used is standard. Similar arguments appeared previously, \emph{e.g.}, in~\cite{bartal96,harpeledmendel06,mendelnaor07,mendel06}.
However, the context here is slightly different, and the designated time bound is ${O}\left(m\log^3 n\right)$, which is faster than the implementations we are aware of.
While the argument is relatively straightforward, a full description of it is tedious to write and read. Instead we only sketch the implementation here.
A complete description, including algorithmic implementation, appears in~\cite{mythesis}.

In a na\"{\i}ve implementation, the number of scales in which we sample CKR partitions  is $\Theta\left(\lg \Phi\right)$. 
This leads to $O((n\log^2 n+m \log n)\log\Phi)$ bound on the expected running time.
Here we develop an implementation having ${O}\left(m \log^3 n\right)$ expected running time. We define for each scale an appropriate quotient of the input graph. We then show that CKR partitions of  those substitutive graph metrics  retain the properties of CKR partitions on original metric. Using those quotients, not all scales need to be processed, and the total size of the quotient graphs in all processed scales is ${O}\left(m \lg n\right)$.

For $y,y'\subseteq X$, let $\rho\left(y,y'\right) = \min \left\{\left. \rho\left(x,x'\right)\ \right| x\in y, x'\in y'\right\}$.
Given a partition $Y$ of the space $(X,\rho)$ we define the quotient metric $\nu$ on $Y$ as
\[ \nu\left(y, y'\right)=\min \Bigl\{ \sum_{j=1}^{l}\rho\left(y_{j-1}, y_j\right)\ : \ y_0,\ldots,y_l\in Y, y_0=y, y_l=y' \Bigr\} . \]

\begin{definition}\label{def:scalequot}
A space $\left(Y,\nu\right)$ is called \emph{$\Delta$-bounded quotient} of an $n$-point metric space $\left(X,\rho\right)$ if $Y$ is a $\Delta$-bounded partition of $X$, $\nu$ is {a} quotient metric on $Y$,  
and for every $x\in X$,  $B_\rho(x,\Delta/n)\subseteq Y(x)$.
\end{definition}
Note that a $\Delta$-bounded quotient of $n$-point metric space exists: define a relation $x\sim x'$ if $\rho(x,x')\le \Delta/n$, and take the transitive closure.
The quotient subsets are the equivalence classes, and by the triangle inequality, the diameter of those equivalence classes is at most $\Delta$.

The following lemma follows easily from Lemma~\ref{lem:improvedCKR}, see the proof of~\cite[Lemma~5]{MN06-arxiv}.
\begin{lemma}\label{lem:aspectratio}
Fix $\Delta>0$, and let $\left(Y,\nu\right)$ be a $\frac{\Delta}{2}$-bounded quotient of $\left(X,\rho\right)$. Let $\sigma:X\to Y$ be the natural projection, assigning each vertex $x\in X$ to its cluster $Y\left(x\right)$.
Let $ L$ be a $(\Delta/2)$-bounded CKR partition of $Y$. 

Let $ P$ be the pullback of $ L$ under $\sigma$, \emph{i.e.}, $ P=\left\{\left.\sigma^{-1}\left(A\right)\ \right|\ A\in  L\right\}$. Then $P$ is a $\Delta$-bounded partition of $X$ such
that for every $0<t\le\Delta/16$ and every $x\in X$,
\begin{equation} \label{eq:aspectratioCKR2}
\Pr\left[B_\rho \left(x,t\right)\subseteq 
P\left(x\right)\right] 
\ge 
\left(\frac{\left|B_\rho\left(x,\Delta/16\right)\right|}{\left|B_\rho\left(x,\Delta\right)\right|}\right)^{\frac{32t}{\Delta}}.
\end{equation}
and furthermore, if $t\le \Delta/2n$,
then
\begin{equation} \label{eq:prob=1}
\Pr\left[B_\rho \left({x},t\right)\subseteq 
P{\left(x\right)}\right] 
= 1.
\end{equation}
\end{lemma}

We define $G|_\Delta$ as the subgraph of $G$ with edges of weight at most $\Delta$ and no isolated vertices. 

\begin{definition} \label{lem:limitconnected}
Given a weighted graph $G=\left(X,E,\omega\right)$ and $\Delta> 0$. Define the graph $G|_\Delta=\left(X|_\Delta , E|_\Delta, \omega|_\Delta\right)$ as follows.
\begin{align*}
E|_\Delta &=\left\{\left(u,v\right)\in E\ :\ u\ne v, \text{ and }\omega\left(u,v\right)\leq \Delta\right\}, \\
X|_\Delta &=\left\{u\in X\ :\ \exists v\in X,\ \left(u,v\right)\in E|_\Delta \right\}, \text{ and }\\
\omega|_\Delta&=\omega|_{E|_\Delta}. 
\end{align*}
\end{definition}

\begin{lemma}\label{lem:irredges}
Given a weighted graph $G=\left(X,E,\omega\right)$, and $\Delta> 0$. 
Let $L$ be a $\Delta$-bounded CKR partition of $X|_\Delta$, using the metric induced by $G|_\Delta$.
Then $ P= L\cup \left\{\left\{v\right\}\ :\ v\in X\setminus (X|_\Delta)\right\}$ is a $\Delta$-bounded CKR partition of $X$, 
using the metric induced by $G$.
\end{lemma}
\begin{proof}
Let $\rho$ be the shortest-path metric on $G$.
Observe that when computing a $\Delta$-bounded CKR partition of $\left(X, \rho\right)$ no edge of weight larger than $\Delta$ is ``used" by the Dijkstra's algorithm for computing the balls, and therefore discarding them does not change the behavior of the algorithm. Also, for each $v\in X\setminus X|_\Delta$, $B_\rho\left(v,\Delta\right)= \left\{v\right\}$, \emph{i.e.}, in any $\Delta$-bounded CKR partition of $X$, $v$ will appear in a singleton subset.
\end{proof}

Lemma~\ref{lem:irredges} and Lemma~\ref{lem:aspectratio} form the basis for dispensing with the dependence on the spread in the construction time. 
We next sketch the scheme we use.
Denote the input graph $G=(X,E,\omega)$, $|X|=n$, $|E|=m$, and let $\rho$ be the graph metric on $G$.

We first
construct an ultrametric $\nu$ on $V$, represented by an ultrametric tree $H =\left(T, \Gamma\right)$ such that for every $u,v\in V$, $\rho\left(u,v\right)\le \nu\left(u,v\right)\le n \rho\left(u,v\right)$.  $H$ can be constructed in $O(m +n \log n)$ time using minimum spanning tree procedure, see~\cite[Section~3.2]{harpeledmendel06}.
For a given
 $\Delta\geq 0$, and a leaf $v\in T$, denote by  $\sigma_{\Delta}\left(v\right)$  the highest ancestor $u$ of $v$ for which $\Gamma\left(u\right)\leq \frac{\Delta}{2n}$. 
Using the level-ancestor data structure (cf.~\cite{Bender2004}) the tree $T$ can be preprocessed in $O(n)$ time
such that queries for $\sigma_\Delta(v)$ (given $\Delta$, and $v$) are answered in $O(\log n)$ time. See~\cite[Section~3.5]{harpeledmendel06} for a similar supporting data structure.
 

Given $\Delta>0$, define the  weighted graph $G_{(\Delta)} $ as follows.
$G_{(\Delta)}=\left(X_{(\Delta)},E_{(\Delta)},\omega_{(\Delta)} \right)$ where,
\begin{align*}
X_{(\Delta)} &=\left\{\sigma_{\Delta}\left(v\right) :\ v\in X\right\},
\\ E_{(\Delta)} &= \left\{\left(\sigma_{\Delta}\left(u\right),\sigma_{\Delta}\left(v\right)\right) :\ \left(u,v\right)\in E, \sigma_{\Delta}\left(u\right) \neq \sigma_{\Delta}\left(v\right) \right\},
\\ \omega_{(\Delta)} \left(u,v\right) &=\min\left\{ \omega\left(w,z\right) :\ \sigma_{\Delta}\left(w\right)=u,\; \sigma_{\Delta}\left(z\right)=v \right\}.
\end{align*}

Let $\rho_{(\Delta)}$ be the shortest-path metric on $G_{(\Delta)}$.
Then, directly from the definitions, $\left(X_{(\Delta)} ,\rho_{(\Delta)}\right)$ is a $\frac{\Delta}{2}$-bounded quotient of $\left(X, \rho\right)$.

For an integer $j\geq -1$  denote
$G_j=\left(V_j,E_j,\omega_j\right)$ where $G_j=\left(G_{(8^j)}\right)|_{8^j/2}$.
The following lemma gives an upper bound on the total size of the graphs $G_j$.
\begin{lemma} \label{lem:G_j}
\[ \sum_{j\ge -1 }\left(\left|V_j\right|+\left|E_j\right|\right) = O\left(m\lg n\right).\]
\end{lemma}

\begin{proof}
Fix $\left(u,v\right)\in E$ {and $j\geq -1$} such that $\left(\sigma_{8^j}\left(u\right), \sigma_{8^j}\left(v\right)\right)\in E_{(8^j)}$. By the definition of $E_{(8^j)}$, 
  $\sigma_{8^j}\left(u\right)\neq \sigma_{8^j}\left(v\right)$. By the definition of $\omega_{(8^j)}$,  $\omega\left(u,v\right)\geq \omega_{(8^j)}\left(u,v\right)\geq \frac{8^j}{2n}$. 
Also, $\left(\sigma_{8^j}\left(u\right), \sigma_{8^j}\left(v\right)\right)\in E_j$ if and only if $\omega_{(8^j)}\left(\sigma_{8^j}\left(u\right), \sigma_{8^j}\left(v\right)\right) \leq \frac{8^j}{2}$. 
So by the triangle inequality $\omega(u,v)\le {\omega_{(8^j)}(\sigma_{8^j}(u),\sigma_{8^j}(v)) +  8^j}$, and hence
$\omega(u,v)\le 1.5\cdot 8^j$.
That is, each edge of $G$ is represented in $G_j$ only when $\omega\left(u,v\right)\in \left[\frac{8^j}{2n}, 1.5\cdot {8^j}\right]$. A total of $O\left(\log n\right)$ scales. By definition, $G_j$ contains only non-isolated vertices, so $\forall j$, $\left|V_j\right|\leq 2 \left|E_j\right|$. 
\end{proof}

Let $\Processed=\{j\ge -1:\; V_j\ne\emptyset\}$. 
\begin{lemma}
The set of graphs $(G_j)_{j\in\Processed}$ can be constructed in $O(m \log ^2 n)$ expected time.%
\footnote{With a bit more care the running time can be improved to $O(m \log n)$. 
This improvement, however, will not improve the total construction time of the hierarchical partition.}
\end{lemma}
\begin{proof}[Sketch of a proof]
First we sort the edges in $E=\{e_1,\ldots e_m\}$ in non increasing order.
Keep a ``sliding window" $[i_L(t),i_R(t)]$, $i_L(t),i_R(t)\in \{1,\ldots,m\}$,
$t\in\{1,\ldots,|\Processed|\}$, 
as follows: Let $j_1=\lceil \log _8 {\Phi}\rceil $. $i_L(1)=1$, $i_R(1)=\max\{i:\; \omega(e_i)\ge 8^{j_1}/2n\}$.
Assuming  $j_{t-1}$ is already defined, define 
$j_t=\max\{j<j_{t-1}:\; \exists i,\; 8^j \ge  \omega(e_i)\ge 8^j/2n\}$,
$i_L(t)=\min\{i:\ \omega(e_i)\le 8^{j_{t}}\}$, and $i_R(t)=\max\{i:\; \omega(e_i)\ge 8^{j_t}/2n\}$. Note that $\{j_t\}_t=\Processed$, and the definition gives $O(m)$ time algorithm for computing the sequences $(j_t)_t$, $(i_L(t))_t$, and $(i_R(t))_t$. Constructing $G_{j_t}$ can now be done in $O((i_R(t)-i_L(t)+1) \log n)$ time, by observing 
that the set of vertices is
\[ V_{j_t}=\{\sigma_{8^{j_t}}(u_i), \sigma_{8^{j_t}}(v_i):\; i\in \{i_L(t),\ldots, i_R(t)\},\; (u_i,v_i)=e_i\}, \]
and similarly the set of edges is
\[ E_{j_t}=\{(\sigma_{8^{j_t}}(u_i), \sigma_{8^{j_t}}(v_i)):\; i\in \{i_L(t),\ldots, i_R(t)\},\; (u_i,v_i)=e_i\}. \]
Another $\log n$ factor in the construction time comes from the $O(\log n)$ time needed for each query of the form ``$\sigma_\Delta(u)$". 
Since $i_R(t)-i_L(t)+1=|E_{j_t}|$, by Lemma~\ref{lem:G_j}, $\sum_t (i_R(t)-i_L(t)+1)=O(m \log n)$.
\end{proof}

Next, we sample $(8^{j_t}/2)$-bounded CKR partition $L_{j_t}$ for each $G_{j_t}$. By Lemma~\ref{lem:aspectratio}, $(L_{j_t})_t$ (implicitly) represents CKR partitions of $G$ in \emph{all} scales.
Using Theorem~\ref{thm:fast-ckr} and Lemma~\ref{lem:G_j} computing $(L_{j_t})_t$  is done in $O(m \log^3 n)$ expected time.

Hierarchical partitions have an $O(n)$ storage representation.
It is similar to an efficient ultrametric tree representation, such as the \emph{nettree}  in~\cite{harpeledmendel06}. Using a rooted tree $\mathcal P$ whose leaves correspond to the points of $X$, each internal vertex $u$ has at least two children, 
and is labeled with a (logarithm of) scale, $s(u)$. The $8^j$-bounded partition $P_j$ is now defined as follows: For $x\in X$, $P_j(x)$ is the 
highest ancestor $u$ of $x$ in $\mathcal P$ such that $s(u)\le j$. 
Since the tree $\mathcal P$ does not have vertices of degree 2, except maybe the root, its size is $O(n)$.

We are left to describe how to compute the the common refinement of the pullbacks of $(L_{j_t})_t$ as a hierarchical partition represented in the tree structure $\mathcal P$ of the previous paragraph.
This is done by top-down fashion as follows:

In the initialization step, $\mathcal P$ is created as a rooted star whose root, $r$ is labeled by $8^{j_1}$, and its leaves correspond to $\{\sigma_{8^{j_1}}(u): \; u\in X\}$.

Next, inductively assume that $\mathcal P$ is a hierarchical partition of $\{\sigma_{8^{j_{t-1}}}(u):\; u\in X\}$ corresponding to $\{L_{j_s}:\; s\le t-1\}$. We refine $\mathcal P$ to include $L_{j_t}$ as follows:
\begin{itemize}
\item Replace the leaves of $\mathcal P$: Each $\sigma_{8^{j_{t-1}}}(u)$ is replaced by 
\[\{\sigma_{8^{j_{t}}}(v):\; v\in X,\  \sigma_{8^{j_{t}}}(v) \text{ is a descendant of }\sigma_{8^{j_{t-1}}}(u) \} .\]
This step is done in $O(|V_{j_t}|)$ time by simply starting from an ``old leaf" $\sigma_{8^{j_{t-1}}}(u)$ as a vertex in $T$ and descending
in $T$ to level $8^{j_t}/2n$.

\item Next, incorporate $L_{j_t}$ into the hierarchical partition in a straightforward way:
Scan the leaves of $\mathcal P$, which are in $V_{j_t}$ grouped by their parents. Fixing such a parent $u$ whose children  $v_1,\ldots, v_\ell$ are all leaves, we partition $v_1,\ldots, v_\ell$ to subsets
$\{\{v_1,\ldots,v_\ell\}\cap C:\; C\in L_{j_t}\}$. For every such subset of size 2 or more we define
a new parent $w$ (which will be a child of $u$) with the label $8^{j_t}$.
\end{itemize}

Hence, the $t$-th iteration in the algorithm above is executed in $O(|V_{j_t}|)$ time,
so the total time for constructing the common refinement is $O(m \log n)$.

\section{Applications}\label{section:applications}


\begin{proof}[Proof of the first part of Theorem~\ref{thm:application-1}]
As observed in Section~\ref{section:preliminaries}, hierarchical partitions correspond
to ultrametrics. As  shown in~\cite{raotalwar04}, when the partition in every scale is a CKR partition, the resulting distribution over ultrametrics is a probabilistic embedding with $O(\log n)$ distortion.%
\footnote{Technically, in~\cite{raotalwar04} the hierarchical partition was built  differently: instead of taking a CKR partition of the whole space in every scale, and then the common refinement, at each scale they took many CKR partitions, one for each block of the partition of the previous scale. This technicality is inconsequential 
for probabilistic embeddings. However, for the construction of approximate distance oracles,
the approach of~\cite{raotalwar04} does not seem to work since it does not have stochastically independent partitions
in the different scales. See also~\cite{mendelnaor07}.}
The algorithm described in Section~\ref{section:optimization} samples a hierarchical partition (and hence an ultrametric) in $O(m \log^3 n)$ expected time.
\end{proof}



\begin{proof}[Proof of the second part of Theorem~\ref{thm:application-1}]
A point $x\in X$ is called $\beta$-padded in hierarchical partition $\mathcal H=(P_{-1},\ldots,P_{\log \Phi})$,
if for every $j$, $B(x,\beta 8^j)\subset P_j(x)$. 

The main part of constructing $O(\beta^{-1})$-approximate distance oracle based on CKR partitions works as follows~\cite{mendelnaor07}: Set $X_0=X$, and iteratively on $i=0,1,\ldots$ do: Compute a hierarchical CKR partition $\mathcal H_i$ of $X_i$, and obtain an ultrametric $H_i$ from $\mathcal H_i$. Let $Y_i$ be a set of $\beta$-padded points in 
$\mathcal H_i$ that is found in Lemma~\ref{lem:ckr-hierarchical-part}. 
Set $X_{i+1}\coloneqq X_i\setminus Y_i$, $i\coloneqq i+1$ and repeat until $X_i=\emptyset$.
The set of ultrametrics $(H_{i})_i$, together with some supporting data-structures constitute the approximate distance oracle. 

By Lemma~\ref{lem:ckr-hierarchical-part}, $\mathbb E |Y_i|\ge |X_i|^{1-32\beta}$,
the number of iterations until $X_i=\emptyset$ is in expectation at most 
$O(\beta^{-1} n^{O(\beta)})$, and hence the total storage is as claimed.

There are two issues in the construction of $(H_i)_i$ that we have not covered yet:
First,  the task is to sample a hierarchical partition of $X_i$ which is only a subset of the vertices in the graph $G=(X,E,\omega)$.
This is rather easy to handle by adapting the algorithms in Section~\ref{section:CKRrandompartition} and Section~\ref{section:optimization} to work with subsets of the vertices.

The second issue is the computation of $\beta$-padded points.
The $\beta$-padded points of a (single) $\Delta$-bounded partition $P$ of a weighted graph $G=(V,E,\omega)$ can be computed as follows: Add a new vertex $s_0$.
For every edge $(u,v)\in E$ such that $P(u)\ne P(v)$, add an edge $(s_0,u)$ whose weight is $\omega(u,v)$.
Execute Dijkstra's shortest path algorithm from $s_0$, and delete all vertices at distance at most $\beta\Delta$ from $s_0$. This can be implemented in $O(m +n\log n)$ time. Note that in the hierarchical partition if a point is not in $V_j$ then it is padded at scale $8^j$. Hence in order to compute a $\beta$ padded point set in hierarchical partition,
for every $t$, we  cross out the points which are not $2\beta$-padded in $L_{j_t}$. The remained points are $\beta$-padded in the pullbacks of $(L_{j_t})_t$ (as follows from Lemma~\ref{lem:aspectratio}) and hence also in the hierarchical partition. When implemented carefully, this can be done
on every graph $G_j$ in $O(|E_j|+|V_j|\log n)$, and by Lemma~\ref{lem:G_j}, in a total $O(m \log ^2 n)$ time.
\end{proof}

A $t$-spanner of a weighted graph $G=(V,E,\omega)$, is a subset of the edges $E'\subset E$  such that 
the shortest-path metric on $(V,E',\omega|_{E'})$ is at most $t$ times the shortest-path metric on $G$. We need the following result.

\begin{theorem}[\cite{baswanasen03}]\label{thm:spanner}
 Let $G=\left(V,E,w\right)$ be a weighted graph with $n$ vertices and $m$ edges, and let $k\geq 1$ be an integer. A $(2k-1)$-spanner of
with $O\left(kn^{1+1/k}\right)$ edges can be computed in $O\left(km\right)$ expected time.
\end{theorem}

\begin{proof}[Proof of Theorem~\ref{thm:hierachical-part}]
By Theorem~\ref{thm:spanner},
given an $n$-point metric space $(X,\rho)$,
a {$5$}-spanner $H$ of $(X,\rho)$ with $O(n^{4/3})$ edges can be constructed in $O(n^2)$ time. 
We next apply the second part of Theorem~\ref{thm:application-1} whose running time is $o(n^2)$.
\end{proof}

\bibliographystyle{plainurl}
\bibliography{ckr}
\end{document}